%% file: main.tex
\documentclass[letterpaper]{article} 
\usepackage{aaai2027}  
\usepackage[hyphens]{url}  
\usepackage{graphicx} 
\usepackage{natbib}  
\usepackage{caption} 
\usepackage{algorithm}
\usepackage{algorithmic}

\usepackage[inline]{enumitem}

\usepackage{amsmath}
\usepackage{cleveref}
\usepackage{amsthm}
    \usepackage{comment}
\usepackage{ltl}

\usepackage{newfloat}
\usepackage{listings}
\DeclareCaptionStyle{ruled}{labelfont=normalfont,labelsep=colon,strut=off} 
\floatstyle{ruled}
\newfloat{listing}{tb}{lst}{}
\floatname{listing}{Listing}

\usepackage{booktabs}

\title{A Hyper-Logical View on Multi-Agent Path Planning}
\author{
    Hadar Frenkel, Dor Atzmon
}
\affiliations{
    Bar-Ilan University

}

\begin{document}
\include{macros}

\maketitle

\begin{abstract}
\emph{Multi-Agent Path Planning} (MAPP) requires collision-free paths for multiple agents. MAPP has numerous real-life applications, including the navigation of automated vehicles, and has therefore many variants considering different constraints, requirements, and objectives. For each such variant, a new algorithmic treatment is often developed. This paper takes a complementary, specification-centered view: we show that many MAPP variants can be modeled in \emph{HyperLTL}, a temporal logic that reasons about multiple paths and their evolution through time, and in its $\omega$-regular extension HyperQPTL. The contribution is not a claim that current hyperproperty model checkers replace specialized MAPP solvers. Rather, the hyper-logical view gives a unified language for comparing MAPP variants, adding cross-agent constraints, expressing optimality, and connecting MAPP to verification, explanation, and causality techniques for hyperproperties.
\end{abstract}

\section{Introduction}
In \emph{Multi-Agent Path Planning} (MAPP), we require planning a set of paths for a set of agents without conflicts (collisions). MAPP is derived from various real-world applications in which physical agents must navigate without colliding with other agents. These applications range from indoor environments, such as robot navigation in automated warehouses~\cite{wolfgang,Li2020LifelongMP}, to planning outdoor routes for autonomous vehicles (e.g., cars, trains, or drones)~\cite{AtzmonDR19,UAV}.

MAPP is general and, due to its broad applicability, has many variants addressing different real-life tasks of planning conflict-free paths. A commonly known MAPP variant is \emph{Multi-Agent Path Finding} (MAPF)~\cite{stern2019multi}, in which each agent's path leads to its given goal. Optimally solving MAPF for standard objective functions is NP-hard~\cite{Surynek10short,YuL13}. Nevertheless, efficient algorithms optimally solve MAPF for many agents~\cite{ZhangLSKK20,LI2021103574,boyarski2022merging,branch-cut_price}. A prominent optimal MAPF algorithm is \emph{Conflict-Based Search} (CBS)~\cite{CBSJUR}. CBS plans a path for each agent separately. Then, repeatedly, it resolves conflicts between the paths by imposing constraints and replanning new paths, until conflict-free paths are found. Besides MAPF, we define eight additional MAPP variants (nine, including MAPF) studied in the literature. For most of them, a specialized algorithm was designed, often by adapting CBS or another MAPF technique to the specific variant.

This paper presents a different view on MAPP problems: instead of proposing another dedicated solver, we show that many MAPP variants can be expressed modularly in the temporal logic HyperLTL~\cite{ClarksonFKMRS14} and its omega-regular counterpart HyperQPTL~\cite{DBLP:phd/dnb/Rabe16}.
HyperLTL is an extension of the prominent logic LTL - Linear Temporal Logic~\cite{Pnueli77}, that can express the behavior of an agent over time, thus allowing us to express properties such as \emph{an agent starts in vertex $s$ and eventually reaches vertex $g$}, but also much more elaborate properties, as we show later in the paper. While LTL can reason about one path at a time, HyperLTL reasons about the relations between different paths, thus, it serves as a natural specification language for MAPP problems. Moreover, it allows us to express not only the requirements posed in the various MAPP instances described below, but we also show that HyperLTL can express additional constraints and behaviors that are not easily captured by CBS (see \Cref{sec:motivation}). 
In \Cref{sec:related}, we cover the related work regarding other connections between planning and temporal (hyper)properties.  
\paragraph{Contributions} 
(1) We present a modular translation of MAPP problems in HyperLTL or HyperQPTL (\Cref{sec:translation}), including examples of cross-agent constraints that are cumbersome to express in decoupled path-planning formalisms (\Cref{sec:motivation}); (2) We show how to express cost functions, and in particular makespan, as hyperproperties (\Cref{sec:makespan}); (3) We reduce finding an optimal solution for the modeled MAPP variants to HyperQPTL (or HyperLTL) model checking (\Cref{sec:MC}); (4) We clarify the intended role of the reduction: it is a solver-independent specification and reasoning layer, not an empirical claim that current hyperproperty tools outperform mature MAPP solvers.

\section{Preliminaries}

\begin{figure}[ht!]
\centering
\includegraphics[width=0.98\linewidth]{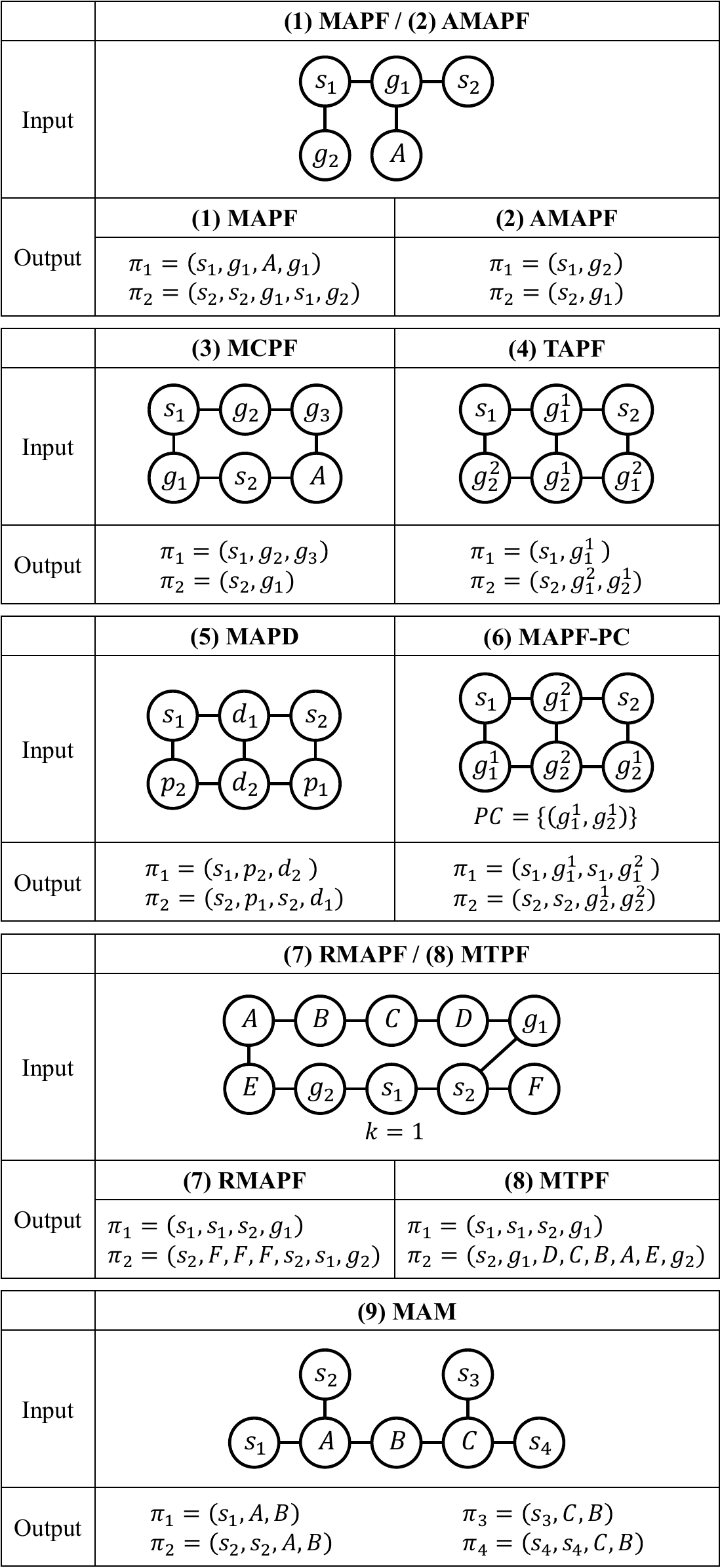}
\caption{Nine MAPP variants/problems.}
\label{fig:problems}
\end{figure}

\subsection{Planning}\label{sec:prelim:planning}
Given an undirected graph $\mathcal{G}=(V,E)$, a \emph{path} $\pi$ for agent $a$ is a list of vertices, where $\pi(t)$ is the vertex at timestep $t$. Between $t$ and $t+1$, agent $a$ either moves to an adjacent vertex ($(\pi(t),\pi(t+1))\in E$) or waits ($\pi(t)=\pi(t+1)$). The path cost is the number of actions in the path: $c(\pi)=|\pi|-1$.

For $n$ agents $A=(a_1,\dots,a_n)$ and start vertices $S=(s_1,\dots,s_n)$, a \emph{plan} $P=(\pi_1,\dots,\pi_n)$ has $\pi_i(0)=s_i$ for each agent $a_i$, and each agent waits at its last vertex forever. Paths $\pi_i$ and $\pi_j$ \emph{conflict} at timestep $t$ if $\pi_i(t)=\pi_j(t)$ (a \emph{vertex conflict}) or between $t$ and $t+1$ if $(\pi_i(t),\pi_i(t+1))=(\pi_j(t+1),\pi_j(t))$ (an \emph{edge conflict}). A MAPP \emph{solution} is a conflict-free plan. An \emph{optimal} MAPP solution minimizes a chosen objective. The standard MAPP objectives are \emph{sum-of-costs} (SOC), which equals the sum of paths' costs:  $C(P)=\sum_i c(\pi_i)$, and \emph{makespan}, which equals the maximum cost among them: $C(P)=\max_i c(\pi_i)$. All nine MAPP variants below use $\mathcal{G},A,S$ and add specific input. Figure~\ref{fig:problems} shows representative instances and optimal 
outputs.

\paragraph{(1) Multi-Agent Path Finding (MAPF)}
~\cite{stern2019multi} receives goals $G=(g_1,\dots,g_n)$ and requires each agent $a_i$ to end at its corresponding goal $g_i$: $\pi_i(|\pi_i|-1)=g_i$.


 
\paragraph{(2) Anonymous MAPF (AMAPF)}~\cite{kloder2006path} has the MAPF input but does not pre-assign goals to agents. 
Each agent reaches some $g_j\in G$: $\pi_i(|\pi_i|-1)=g_j$.

\paragraph{(3) Multi-Agent Combinatorial Path Finding (MCPF)}
~\cite{ren2023cbss} extends AMAPF to $m$ goals $G=(g_1,\dots,g_m)$, where agents may satisfy multiple goals or none. A solution requires every goal $g_j\in G$ to appear on some path $\pi_i$: $g_j \in \pi_i$.


\paragraph{(4) Combined Target-Assignment and Path Finding (TAPF)}~\cite{honig2018conflict} gives each agent a goal set. 
For $GSET=(gset_1,\dots,gset_n)$, each agent $a_i$ reaches one vertex in $gset_i=\{g_i^1,g_i^2,\dots\}$: $\pi_i(|\pi_i|-1)\in gset_i$.


\paragraph{(5) Multi-Agent Pickup and Delivery (MAPD)}~\cite{MAPD} requires fulfilling tasks, each with a pickup and delivery vertex. 
Here, $Tasks=(task_1,\dots,task_m)$ and $task_j=\tuple{p_j,d_j}$.
\footnote{MAPD is often a lifelong problem. We consider the \emph{one-shot} version; new task arrivals can be handled by re-solving.} Each task is fulfilled by some agent reaching $p_j$ at $t_1$ and $d_j$ at $t_2>t_1$: 
$\pi_i(t_1)=p_j \wedge \pi_i(t_2)=d_j \wedge t_1<t_2$.



\paragraph{(6) MAPF for Precedence-Constrained Goal Sequences (MAPF-PC)}~\cite{zhang2022multi} receives goal sequences $GSEQ=(gseq_1,\dots,gseq_n)$ and precedence constraints $PC=\{pc^{xy}_{ij},\dots\}$. 
 Let $t_i^x$ be the timestep in which the agent $a_i$ visits the goal $g_i^x$ in a given solution. 
 Each goal sequence $gseq_i=(g_i^1,g_i^2,\dots)$ must be visited by agent $a_i$ in order; that is, for any two goals $g_i^x,g_i^y\in gseq_i$ such that $x<y$, agent $a_i$ must visit goal $g_i^x$ before goal $g_i^y$ ($t_i^x<t_i^y$). 
 Each precedence constraint $pc^{xy}_{ij}=(g_i^x,g_j^y)$ requires agent $a_i$ to visit goal $g_i^x$ before agent $a_j$ visits goal $g_j^y$ ($t_i^x<t_j^y$).




\paragraph{(7) Robust MAPF (RMAPF)} 
~\cite{atzmon2020robust,chen2021symmetry} adds a delay parameter $k\geq 0$ to MAPF. A solution reaches the usual MAPF goals and satisfies $\pi_i(t_1)\neq\pi_j(t_2)$ for all $i\neq j$ whenever $|t_1-t_2|\leq k$.



\paragraph{(8) Multi-Train Path Finding (MTPF)}
~\cite{AtzmonDR19,chen2022multi} uses the RMAPF input, but each agent occupies its current vertex (head) and the previous $k$ vertices (tail). 
Agents must avoid conflicts with both heads and tails; hence, a RMAPF solution may fail for MTPF (Figure~\ref{fig:problems}).

\paragraph{(9) Multi-Agent Meeting (MAM)}
~\cite{atzmon2023conflict} has no goals; all agents must end at some meeting vertex $v_m$: $\pi_i(|\pi_i|-1)=v_m$. The agents do not conflict at vertex $v_m$.

\smallskip

AMAPF and MAM can be optimally solved in polynomial time. All other above variants are NP-hard to solve optimally.


\subsection{Logic}\label{sec:prelim:logic}
We use the notation $[n]$ to denote the set $\{ i ~|~ 1\leq i \leq n\}$.
\paragraph{Transition Systems and Traces}
Let ${AP}$ be a finite set of atomic propositions (in our setting, the atomic propositions will denote the agents' locations). 
A \emph{labeled transition system} (LTS) $T = (Q, Q_0, \delta, L)$ is a tuple where $Q$ is a finite set of states, $Q_0\subseteq Q$ is a set of initial states, $\delta \subseteq Q\times Q$ is a transition relation, and $L: Q \rightarrow 2^{AP}$ is a labeling function, that labels each state with a subset of atomic propositions. 

A \emph{trace} $L(q_0) L(q_1) L(q_2) \cdots \in (2^{AP})^\omega$ of $T$ is an infinite sequence with $q_0\in Q_0$ and $(q_t,q_{t+1})\in\delta$ for all $t\geq 0$.\footnote{We use the standard infinite-word semantics, which coincide with finite-word semantics for the problems studied here.} 
We denote the set of all traces of a transition system $T$ as $\TR(T)$.

\paragraph{HyperLTL}
HyperLTL~\cite{ClarksonFKMRS14} extends LTL~\cite{Pnueli77} with trace quantification, enabling relational properties over executions. Given atomic propositions $AP$ and trace variables $\mathcal{V}$, its syntax is as follows; $q$ and $\exists q$ are used only later for HyperQPTL:
%

\begin{equation}
\begin{aligned}
  &\psi := p_\pi ~|~ \psi \vee \psi ~|~ \neg\psi ~|~ \LTLnext \psi ~|~ \psi\LTLuntil \psi ~|~ q\\
  &\varphi := \exists \pi. \varphi ~|~ \forall \pi. \varphi ~|~ \psi ~|~ \exists q. \varphi 
\end{aligned}
\tag{$\star$}\label{eq:hyperLTL}
\end{equation}
Here $p\in AP$; $\pi\in\mathcal{V}$; and $\LTLnext,\LTLuntil$ are the temporal operators \emph{next} and \emph{until}; we use the usual Boolean abbreviations, and the temporal derivations $\LTLeventually \psi=\mathit{true}\LTLuntil\psi$ (\emph{eventually}), and $\LTLglobally\psi=\neg\LTLeventually\neg\psi$ (\emph{globally}). Intuitively, $p_\pi$ refers to proposition $p$ on trace $\pi$, and $\exists\pi$/$\forall\pi$ quantify over system traces. 
\begin{example}
    In Problem (1) of \Cref{fig:problems}, trace $\pi_1$ satisfies $s_1$ at time $0$, so $\pi_1 \vDash s_1$; also $\pi_1 \vDash \LTLnext g_1 \wedge \LTLeventually A$, 
     since $g_1$ holds at the second timestep, and $A$ eventually holds on $\pi_1$. In addition, 
 $\pi_2 \vDash \LTLglobally \neg A$
since $A$ never occurs on $\pi_2$. 
\end{example}

The semantics are given w.r.t. a trace-set $\mathcal{T}$ and an assignment $\Pi:\mathcal{V}\rightarrow\mathcal{T}$. Let $\Pi[\pi\mapsto\tau]$ update $\pi$ to $\tau$, and let $\Pi[t,\infty)$ map every variable to the suffix of its trace from position $t$. Then:
\begin{align*}
    &\Pi, \mathcal{T} \vDash p_\pi &\text{ iff } & p\in\Pi(\pi)(0) \\ 
    &\Pi, \mathcal{T} \vDash \neg \psi &\text{ iff } &\Pi,\mathcal{T} \nvDash \psi \\
    & \Pi, \mathcal{T} \vDash \psi_1 \vee \psi_2  &\text{ iff } &  \Pi, \mathcal{T} \vDash \psi_1 \text{ or } \Pi, \mathcal{T} \vDash \psi_2 \\
       &\Pi, \mathcal{T} \vDash \LTLnext \psi &\text{ iff } & \Pi[1,\infty),\mathcal{T} \vDash \psi \\ 
      &\Pi, \mathcal{T} \vDash \psi_1 \LTLuntil \psi_2 &\text{ iff } & \exists t_2\in\mathbb{N} \text{ s.t. } \Pi[t_2,\infty),\mathcal{T} \vDash \psi_2 \\ &&&\text{ and } \forall t_1< t_2: \Pi[t_1,\infty),\mathcal{T} \vDash \psi_1 \\ 
      &\Pi, \mathcal{T} \vDash \exists \pi. \varphi &\text{ iff } & \text{ there exists } \tau\in\mathcal{T} \text{ s.t. } \Pi[\pi\mapsto \tau],\mathcal{T}\vDash \varphi \\
       &\Pi, \mathcal{T} \vDash \forall \pi. \varphi &\text{ iff } & \text{ forall } \tau\in\mathcal{T}: \Pi[\pi\mapsto \tau],\mathcal{T}\vDash \varphi
\end{align*}
We say that $\mathcal{T} $ satisfies $\psi$, denoted $\mathcal{T} \vDash \psi$, if $\Pi_\emptyset, \mathcal{T} \vDash \psi$ where $\Pi_\emptyset: \mathcal{V}\rightarrow \emptyset$ is the empty assignment, meaning that $\psi$ contains no free trace variables that need to be assigned.
We say that a transition system $T$ satisfies a HyperLTL formula~$\psi$ if $\TR(T) \vDash\psi$.

\paragraph{HyperQPTL} 
HyperQPTL~\cite{DBLP:phd/dnb/Rabe16} extends HyperLTL with 
quantification over {atomic propositions}, to express $\omega$-regular hyperproperties, similar to the omega-regular extension of QPTL~\cite{QPTL} over LTL~\cite{LTL}. 

Formulas are defined over $AP\cup AP'$, where $AP'$ contains quantified propositions disjoint from $AP$. We define the semantics of 
$q$ and $\exists q$ from $(\star)$, 
for $q\in AP'$, using a fresh trace variable $\pi_f$:
\begin{align*}
          &\Pi, \mathcal{T} \vDash q &\text{ iff } & q\in\Pi(\pi_f)(0) \\
          &\Pi, \mathcal{T} \vDash \exists q. \varphi &\text{ iff } & \text{ exists } \tau\in (2^{AP'})^\omega \text{ s.t. } \Pi[\pi_f\mapsto \tau],\mathcal{T}\vDash \varphi 
\end{align*}
Quantified propositions are global auxiliary labels; for example, they can mark the timestep at which an agent satisfies its goal. In Section~\ref{sec:modeling}, we use HyperLTL unless stated otherwise.\footnote{HyperQPTL does not add computational complexity~\cite{DBLP:phd/dnb/Rabe16}; Since HyperLTL is a more common formalism, we choose to use it when possible.}
\begin{example}\label{ex:qptl}
In Problem (1) of \Cref{fig:problems}, $a_1$ visits $g_1$ on the second timestep but satisfies it only on the fourth, after clearing it for $a_2$. We mark this satisfying visit by a quantified proposition $q_1$, enforced to occur once:
\begin{align*}
    \exists q_1. \big(\neg q_1 \LTLuntil (q_1 \wedge \LTLnext \LTLglobally \neg q_1) \big) \wedge \LTLeventually (q_1 \wedge {g_1}_{\pi_1})
\end{align*}
The first conjunct states that $q_1$ appears exactly once; the second states that this unique mark coincides with $g_1$ on $\pi_1$.
\end{example}

\paragraph{Syntactic Sugar}
We use $\pi_1=\pi_2$ for $\LTLglobally \bigwedge_{p\in AP}(p_{\pi_1}\leftrightarrow p_{\pi_2})$ and $\pi_1\neq\pi_2$ for $\LTLeventually \bigvee_{p\in AP}(p_{\pi_1}\leftrightarrow\neg p_{\pi_2})$.
We use $\LTLnext^i$ for $i$ applications of the \emph{next} operator.


\section{Modeling MAPP in HyperLTL}\label{sec:modeling}
In this section, we present simple and modular translations of the various 
 MAPP problems into respective HyperLTL formulas. 
We separate between the core planning problem and optimality w.r.t. the cost function, allowing modularity both between different problems and different cost functions. We discuss translations of cost functions in Section~\ref{sec:cost}. 

\subsection{Motivation} \label{sec:motivation}
The main goal of this paper is to make the logical structure of MAPP variants explicit. A hyper-logical specification separates what is common across variants (start conditions, goal satisfaction, and collision avoidance) from what changes between them. Thus, moving from one variant to another often amounts to replacing a local component of the specification rather than redesigning an algorithm from scratch. This is useful even before considering tool performance: it gives a concise way to compare variants, expose hidden assumptions, identify redundant constraints, and state new requirements without committing to a particular solver architecture.

This scope also explains why a large empirical comparison is not the right central test for the present contribution. Current HyperLTL and HyperQPTL model checkers are general hyperproperty-verification tools, and their worst-case complexity is high~\cite{CSL2020complexity,DBLP:phd/dnb/Rabe16}. We therefore do not claim that existing tools compete with mature MAPP solvers on standard benchmarks. Such an experiment would primarily measure the present maturity of hyperproperty model checking, rather than the value of the specification layer. The intended contribution is instead a formal bridge: once a MAPP variant is written as a decidable hyperproperty, existing and future verification procedures can be applied to it, and the formula itself can serve as a solver-independent reference specification.

Another key point in our modeling is that HyperLTL directly expresses relations between different paths. This obviously allows us to express MAPP problems, which inherently require reasoning about the paths of all agents; but it also allows us to express additional requirements on the paths of the agents, which are usually not easy to express by decoupled MAPP algorithms, like CBS, which repeatedly plans paths for single agents. For example, we can specify that the path of agent $a_1$ is always close to the path of agent $a_2$, by defining a set of neighbor vertices $N(v) = \{ v' ~|~ (v,v')\in E \}$, and defining the following HyperLTL constraint, where $\pi_i$ is the path of agent $a_i$ for $i=1,2$: 
\begin{align} \label{eq:cbs}
    \LTLglobally \bigwedge_{v\in V} (v_{\pi_1} \rightarrow \bigvee_{v'\in N(v)} v'_{\pi_2})
\end{align}
Through its whole path, if agent $a_1$ is in vertex~$v$, agent $a_2$ has to be in one of its neighboring vertices. 
We can conjunct \Cref{eq:cbs} with any of the formulas in \Cref{sec:translation}, to add this constraint to any of the problems. 
This is only one example of how the modeling as a formal specification allows us to easily add constraints that are not natural to solve otherwise. Similarly, reasoning patterns used inside MAPP algorithms, such as pairwise mutex reasoning or rectangle-style exclusions, can be described as relations among two or more traces under the constraints of a search node. The point is not that HyperLTL replaces these algorithmic ideas, but that it provides a common language in which pairwise and beyond-pairwise path relations can be stated, compared, and eventually verified.

\subsection{{Modeling the MAPP Instance as an LTS}}\label{sec:lts}
Since we evaluate HyperLTL formulas over LTSs, we first transform the MAPP instance into an LTS. 
The transformation is pretty straightforward, and we mention it here for completeness.
Let $\mathcal{G},A,S$ be a MAPP instance, with $\mathcal{G}=(V,E)$, $A=(a_1,\dots,a_n)$ and $S=(s_1,\dots,s_n)$. We denote $AP = V$, that is, the set of location vertices, and
we create an LTS $T_\mathcal{G} = (V, V, \delta_E, L)$
where $\delta_E = \{ (v, v')~|~ v=v' \vee (v,v') \in E \vee (v', v)\in E\}$, and $L(v) = \{ v\}$. The disjunct $v=v'$ captures wait actions. The set $\TR(T_\mathcal{G})$ contains only traces over singleton sets, which denote the location in the graph at each time point. Therefore, for simplicity, from now on we refer to traces over $AP^\omega$ rather than $(2^{AP})^\omega$.
\begin{observation}\label{obs:paths}
    The set $\TR(T_\mathcal{G})$ is equal to the set of all possible paths in the MAPP instance. 
\end{observation}
\subsection{Translation into HyperLTL}\label{sec:translation}
We now demonstrate how to translate different MAPP problems into HyperLTL and its variants. 
We use the notation $\varphi(\overline{\pi})$ to denote a formula $\varphi$ with free variables $\pi_1,\ldots,\pi_n$
(and similarly, $\varphi(\overline{\pi'})$ if the free variables are $\pi_1',\ldots,\pi_n'$, etc.). 
For all of the following problems, we use the classical assumption that once an agent reaches its goal, it stays there forever. 
At the end of this section, we discuss how we can express the problems under the assumption that after the agent reaches its goal, it clears it and can move freely. 

\paragraph{(1) MAPF} 
We start by modeling the basic MAPF problem with $n$ agents, as presented in Section~\ref{sec:prelim:planning}, as an HyperLTL formula:\footnote{Proper HyperLTL formulas have no free variables, we address this in Section~\ref{sec:MC}.} We recall that the basic problem defines a start vertex and a goal vertex for each agent, and requires that each agent reaches from its start vertex to its goal, without conflicts on the paths of agents. Then, we define: 
\begin{equation*}
   \varphi_{\text{MAPF}}(\overline{\pi}):=
  \varphi_{\text{start}}(\overline{\pi}) \wedge \varphi_{\text{goal}}(\overline{\pi}) \wedge \varphi_{\text{no-conf}}(\overline{\pi})
\end{equation*}
Where: 
\begin{equation*}
   \varphi_{\text{start}}(\overline{\pi}) := \bigwedge_{i\in[n]} {s_i}_{\pi_i}  \quad \quad
   \varphi_{\text{goal}}(\overline{\pi}) := \bigwedge_{i\in[n]} \LTLeventually \LTLglobally {g_i}_{\pi_i}
   \end{equation*}
   \begin{equation*}
   \varphi_{\text{no-conf}} (\overline{\pi}):= \LTLglobally \bigwedge_{i\neq j} \bigwedge_{v\in AP}\neg (v_{\pi_i} \wedge v_{\pi_j})
\end{equation*}
Intuitively, $\varphi_{\text{MAPF}}$ requires the existence of a path $\pi_i$ for each agent $a_i$, that starts from $a_i$'s start vertex (c.f. $\varphi_{\text{start}}$), ends at its goal vertex and stays there forever ($\varphi_{\text{goal}}$), and such that all paths are non-conflicting ($\varphi_{\text{no-conf}}$). $\varphi_{\text{MAPF}}$ refers to vertex conflict, but we can easily change it to express edge conflict by changing $\varphi_{\text{no-conf}}$ to $\varphi_{\text{no-conf}}'$ as follows: 
   \begin{equation*}
   \varphi_{\text{no-conf}}'(\overline{\pi}) := \LTLglobally \bigwedge_{i\neq j} \bigwedge_{v,v'\in AP}\neg \big((v_{\pi_i} \wedge v'_{\pi_j}) \wedge \LTLnext (v'_{\pi_i} \wedge v_{\pi_j}) \big)
\end{equation*} 
We use formulas $\varphi_{\text{start}}, \varphi_{\text{goal}},\varphi_{\text{no-conf}}$ throughout the paper to model the various problems, demonstrating the modularity of using HyperLTL as an underlying specification language. 
\begin{theorem}\label{thm:mapf}
    Let $\overline{\pi} = (\pi_1, \ldots, \pi_n )$, let $\tuple{\mathcal{G},A,S,G}$ be an MAPF instance, 
    and let $T_\mathcal{G}$ be a transition system as described in \Cref{sec:lts}. 
    Then, $T_\mathcal{G}$ satisfies $\varphi_{\text{MAPF}}(\overline{\pi})$, iff  $\pi_1, \ldots, \pi_n$ represent a solution to the MAPF problem. 
\end{theorem}

\begin{proof}
We use Observation~\ref{obs:paths}: 
If $\varphi_{\text{MAPF}}(\overline{\pi})$ holds in $T_\mathcal{G}$, we have, directly from the definition of $\varphi_{\text{MAPF}}$, that  $\pi_i, i\in[n]$ is a valid path for agent $a_i$: It starts from $s_i$ and eventually arrives at $g_i$ (and from then, $g_i$ holds globally, that is, the agent stays at $g_i$ forever). Due to $\varphi_{\text{no-conf}}$, for $i\neq j$, paths $\pi_i$ and $\pi_j$ cannot be at the same vertex $v$ at the same time, i.e., all paths are conflict-free. To conclude, we have that $\pi_1, \cdots, \pi_n$ represent valid conflict-free paths for agents $a_1, \cdots, a_n$. 

In the other direction, if there is a solution to the MAPF problem, then 
there are paths $\pi_1, \cdots, \pi_n$ satisfying the formulas $\varphi_{\text{start}}, \varphi_{\text{goal}}, \varphi_{\text{no-conf}}$ (again, due to Observation~\ref{obs:paths}). In particular, we note that $\varphi_{\text{goal}}$ requires that the agents \emph{eventually} arrive at the goal and remain there, however, they are allowed to pass through it before, and continue to move around, to allow other agents to satisfy their goal. 
\end{proof}
\paragraph{(2) AMAPF}
%
%
In the AMAPF problem, we do not pre-assign goals to agents. Thus, 
we define: 
\begin{align}
    \varphi_{\text{AMAPF}}(\overline{\pi}) := &
  \exists \pi_1' \ldots \exists \pi_n'. ~ \varphi_{\text{start}} (\overline{\pi}) \wedge \varphi_{\text{goal}} (\overline{\pi'}) ~\wedge
\\
&  \varphi_{\text{no-conf}}(\overline{\pi}) \wedge 
\\
  & \forall \tau .\exists \tau'. \bigwedge_{i\in[n]} \Big( (\tau = \pi_i ) \rightarrow\label{line:tau0} \\ & \bigvee_{j\in [n]} \big( (\tau' = \pi_j' ) \wedge (\tau = \tau')\big) \Big)
  \label{line:tau}
\end{align}
\emph{Correctness.} The traces $\pi_i$ start at the prescribed starts and are non-conflicting. The traces $\pi'_j$ reach the goals, and Equations~(\ref{line:tau0})-(\ref{line:tau}) require every $\pi_i$ to equal some $\pi'_j$. Since the $\pi_i$ traces are distinct and there are $n$ traces of each kind, this induces a bijective start-to-goal assignment. Thus, a separate $\varphi_{\text{no-conf}}(\overline{\pi'})$ conjunct is redundant and is intentionally omitted.

\paragraph{(3) MCPF} 
Here, we get a list of not pre-assigned goals to satisfy, where one agent can satisfy several goals, or no goals. Then, we define the formulas $\varphi_{\text{goal-MCPF}}$ and $\varphi_{\text{MCPF}}$:
$$\varphi_{\text{goal-MCPF}}(\overline{\pi}) := \bigwedge_{j\in[m]} \bigvee_{i\in[n]} \LTLeventually  {g_j}_{\pi_i} $$
\begin{equation*}
    \varphi_{\text{MCPF}}(\overline{\pi}) :=  
  \varphi_{\text{start}} (\overline{\pi}) \wedge \varphi_{\text{goal-MCPF}}(\overline{\pi}) 
  \wedge \varphi_{\text{no-conf}} (\overline{\pi})
\end{equation*}
\emph{Correctness.} $\varphi_{\text{goal-MCPF}}$ states that each goal is eventually reached by some agent; the usual start and no-conflict conjuncts give the remaining requirements.
\paragraph{(4) TAPF} 
Here, each agent has several goals, and it must reach one of them. We define $\varphi_{\text{TAPF}}$ and $\varphi_{\text{goal-TAPF}} $: 
$$ \varphi_{\text{goal-TAPF}} (\overline{\pi})
:= \bigwedge_{i\in[n]} \bigvee_{g \in gset_i} \LTLeventually \LTLglobally {g}_{\pi_i}$$
\begin{equation*}
    \varphi_{\text{TAPF}}(\overline{\pi}) :=  
  \varphi_{\text{start}} (\overline{\pi})\wedge \varphi_{\text{goal-TAPF}}(\overline{\pi})
  \wedge \varphi_{\text{no-conf}} (\overline{\pi})
\end{equation*}
\emph{Correctness.} $\varphi_{\text{goal-TAPF}}$ states that each agent reaches and remains at one of its allowed goals; the usual start and no-conflict conjuncts complete the encoding.
AMAPF can also be encoded through the TAPF pattern by taking $gset_i=G$ for every agent $a_i$ (under the standard assumption that agents remain at their assigned goals, no-conflict forces different agents to settle at different goals). We keep the separate AMAPF encoding because it makes the anonymous assignment explicit through equality of start-side and goal-side traces. This illustrates a modeling choice rather than a different computational assumption.
\paragraph{(5) MAPD} 
Here, we have a set of pickup-delivery assignments, each should be performed by some agent. Then, similarly to MCPF, we define the formulas $ \varphi_\text{goal-MAPD}$ and $ \varphi_\text{MAPD}$: 
\begin{equation*}
    \varphi_\text{goal-MAPD}(\overline{\pi}):= \bigwedge_{j\in[m]}\bigvee_{i\in[n]}\LTLeventually ({p_j}_{\pi_i} \wedge \LTLeventually {d_j}_{\pi_i})
\end{equation*}
%
\begin{equation*}
    \varphi_{\text{MAPD}}(\overline{\pi}) :=  
    \varphi_{\text{start}} (\overline{\pi}) \wedge \varphi_\text{goal-MAPD}(\overline{\pi})   \wedge \varphi_{\text{no-conf}} (\overline{\pi})
\end{equation*}
\emph{Correctness.} $\varphi_\text{goal-MAPD}$ requires each task $\langle p_j,d_j\rangle$ to be assigned to some agent that eventually reaches $p_j$ and later $d_j$; $\varphi_\text{start}$ and $\varphi_\text{no-conf}$ provide the standard path constraints.


\paragraph{(6) MAPF-PC} 
To express the precedence constraints, we need to use HyperQPTL (Section~\ref{sec:prelim:logic}).
Let $AP' = \{ q_i ~|~ i\in[n] \}$, i.e., we add a quantified proposition for each agent. 
Let $gseq_i= (g_i^1, \cdots, g_i^{m_i})$.
For each~$i$ we define the formula $\psi_i$, that states that $q_i$ occurs exactly $m_i$ times on a trace:
\begin{align}\label{line:PC1}
  \psi_i:= \overbrace{  \neg q_i \LTLuntil (q_i \wedge \LTLnext (\neg q_i \LTLuntil ( q_i \wedge \LTLnext (\cdots  \wedge \LTLnext(\neg q_i \LTLuntil (q_i }^{\neg q_i\LTLuntil q_i \text{ appears }  m_i \text{ times}}\\ \wedge \LTLnext \LTLglobally \neg q_i)\cdots
 )\label{line:PC2}
\end{align}
Next, 
for each agent $a_i$, we define the formula $\varphi_{\text{gseq}}(\pi_i)$, stating that goals are reached in the order specified by $gseq_i$. Since generally, an agent is allowed to visit some vertex $g_j$ on its way to satisfy a goal $g_i$ for $i <j$, we use the quantified proposition $q_i$ to indicate the timestep in which the agent visits the vertex when \emph{satisfying} the goal. Due to $\psi_i$, an agent can satisfy exactly $m_i = | gseq_i|$ goals. 
Thus we define: 
\begin{align*}
 \varphi_{\text{gseq}}(\pi_i) :=~ &{s_i}_{\pi_i} \wedge   \LTLeventually (({g_i^1}_{\pi_i}\wedge q_i) \wedge \\ &\LTLeventually ( ({g_i^2}_{\pi_i}\wedge q_i) \wedge \LTLeventually ( \cdots \wedge \LTLeventually ({g_i^{m_i}}_{\pi_i}\wedge q_i))\cdots)    
\end{align*}
For each constraint 
$(g_i^x, g_j^y)$ 
we define the formula: 
\begin{align*}
    \varphi_{(g_i^x, g_j^y)}(\pi_i, \pi_j) := \LTLeventually \big(( {g_i^x}_{\pi_i} \wedge q_i ) \wedge  \LTLeventually ( {g_j^y}_{\pi_j} \wedge q_j )\big) 
\end{align*}
which states the precedence between two satisfied (expressed by the $q_i, q_j$ propositions) goals. 

We are now ready to define a HyperQPTL formula to capture the whole problem, as follows:
\begin{align*}
    &\varphi_{\text{MAPF-PC}}(\overline{\pi}) := \exists q_1 \ldots \exists q_n. ~ 
    \bigwedge_{i\in [n]} \psi_i ~\wedge ~\varphi_{\text{no-conf}}  (\overline{\pi})~\wedge 
     \\ &\bigwedge_{i\in [n]} \varphi_{\text{gseq}}(\pi_i) ~\wedge
    \bigwedge_{(g_i^x, g_j^y) \in PC}  \varphi_{(g_i^x, g_j^y)}(\pi_i, \pi_j)
\end{align*}
\emph{Correctness.} $\varphi_{\text{gseq}}(\pi_i)$ enforces the order of the goals in $gseq_i$, while $q_i$ marks the timesteps at which goals are satisfied rather than merely visited. The formulas $\varphi_{(g_i^x,g_j^y)}$ then compare these marked timesteps across agents, so a satisfying assignment to the $q_i$ propositions is exactly a choice of goal-satisfaction times that respects all precedence constraints.

\paragraph{(7) RMAPF}
Here, we add to the standard MAPF problem the constraint that the paths of two agents cannot collide over $k$ timesteps (and not only the current timestep). 
Thus, to express the RMAPF problem, we only need to redefine the formula expressing (vertex) conflict, as follows:
\begin{align*}
    \varphi_{\text{R-no-conf}} (\overline{\pi}) := 
\LTLglobally \bigwedge_{i\neq j} \bigwedge_{v\in AP}\neg \big(v_{\pi_i} \wedge ( \bigvee_{t\leq k } \LTLnext^t v_{\pi_j})\big)
\end{align*}

We use the formulas for start and goal vertex as in the MAPF problem, and thus we can express RMAPF as: 
\begin{align*}
        \varphi_{\text{RMAPF}} (\overline{\pi}) :=  \varphi_{\text{start}} (\overline{\pi})  ~\wedge~     \varphi_{\text{goal}} (\overline{\pi})~ \wedge  ~ \varphi_{\text{R-no-conf}} (\overline{\pi}) 
\end{align*}
\emph{Correctness.} $\varphi_{\text{R-no-conf}}$ forbids every other agent from occupying the same vertex within the next $k$ timesteps; start and goal satisfaction are unchanged from MAPF.
\paragraph{(8) MTPF} 
Here, each agent occupies $k~+~1$ vertices (their head, and $k$ vertices of tail). Therefore, to avoid conflicts, two agents cannot share their last $k+1$ vertices. To express this, we need to reason not only about the change over time, but also the change in \emph{locations} over time. 
We therefore modify the definition of conflict as follows:
%

\begin{align}
    &\varphi_{\text{MTPF-no-conf}} (\overline{\pi}): =  ~  \varphi_{\text{no-conf}} (\overline{\pi}) ~\wedge \label{line:no-conf} \\
    \label{line:train1}
    & \LTLglobally 
    \bigwedge_{i \in[n]}~ \bigwedge_{\substack{v_1,\ldots,v_{k+1}\in AP
    \\ \forall i\in[k]: v_i \neq v_{i+1}}} \bigg( \Big({v_1}_{\pi_i} \LTLuntil \big({v_2}_{\pi_i} \LTLuntil ( \cdots  \\   \label{line:train2}
    & \LTLuntil ({v_k}_{\pi_i} \LTLuntil  {v_{k+1}}_{\pi_i}) \cdots )
    \big) \Big)  \rightarrow   \\ \label{line:train3}
    & 
    \Big({v_1}_{\pi_i} \LTLuntil \big( ({v_2}_{\pi_i} \bigwedge_{j\neq i}\bigwedge_{t<2} \neg{v_t}_{\pi_j}  ) )\LTLuntil ( \cdots  \\ 
    \label{line:train4}& \LTLuntil ( ({v_k}_{\pi_i} \wedge \bigwedge_{j\neq i} \bigwedge_{t<k} \neg{v_t}_{\pi_j} ) )\LTLuntil  {v_{k+1}}_{\pi_i}) \cdots )
    \big) \Big) 
\end{align}

Again, the start and goal formulas are as before, and we can express the whole problem as:
\begin{align*}
    \varphi_{\text{MTPF}}  (\overline{\pi}): = 
    \varphi_{\text{start}} (\overline{\pi})  ~\wedge~     \varphi_{\text{goal}} (\overline{\pi})~ \wedge  ~ 
    \varphi_{\text{MTPF-no-conf}} (\overline{\pi})
\end{align*}

\emph{Correctness.} \Cref{line:no-conf} gives ordinary no-conflict. 
     The remaining conjunct specifies that any sequence of $k$ vertices occupied by some agent (Equations~(\ref{line:train1})-(\ref{line:train2})) cannot conflict, at any point in the sequence, with any other agent (Equations~(\ref{line:train3})-(\ref{line:train4})).
     The operator $\LTLuntil$, appearing $k$ times, tracks the change in \emph{locations}, rather than time (to track $k$ timesteps, it would have been sufficient to use the $\LTLnext$ operator).

\paragraph{(9) MAM}  
Here, we need to find a meeting point for all agents. Thus, we define: 
%
\begin{align}
  &  \varphi_{\text{MAM}}(\overline{\pi}) :=  
    ~\varphi_{\text{start}}(\overline{\pi}) ~\wedge ~ \\ &\bigvee_{v\in AP} \Big(\big( \LTLeventually \bigwedge_{i\in [n]} v_{\pi_i} \big) ~\wedge~ \label{line:meeting} \\ &
     \LTLglobally
    \big( \bigwedge_{i\neq j} \bigwedge_{\substack{v'\in AP \\ v' \neq v}} 
    \neg (v'_{\pi_i} \wedge v'_{\pi_j}) \big) \Big) \label{line:meetingconf}
\end{align}
\emph{Correctness.} \Cref{line:meeting} chooses a vertex where all agents eventually meet, while \Cref{line:meetingconf} forbids collisions at every other vertex.

\subsubsection{Other Assumptions on Reaching the Goal}\label{sec:reachesgoal}
Above, we have assumed that once an agent reaches its goal, it stays there, as stated in \Cref{sec:prelim:planning} (hence, we used $\LTLeventually \LTLglobally g$ for the \emph{goal} formulas). However, this can be easily changed to express other assumptions on reaching the goal. For example, if we assume that agent $a_i$ clears the goal vertex and continues to move in the plan, we simply change $\LTLeventually \LTLglobally g_{\pi_i}$ in $\varphi_{\text{goal}} $ to $\LTLeventually g_{\pi_i}$. This can be modified to other assumptions as well, demonstrating the usefulness of this modular modeling. 


\section{Cost Functions as Hyperproperties}\label{sec:cost}
Computing the cost of a solution is orthogonal to finding the solution itself. 
We focus on \emph{makespan}, as the \emph{sum-of-costs} (SOC) requires asynchronous reasoning which is out of the scope of this paper. We discuss SOC shortly in \Cref{sec:SOC}.

\subsection{Makespan} \label{sec:makespan}
A solution $\overline{\pi} = (\pi_1, \cdots, \pi_n)$ is more optimal than a solution  $\overline{\pi}' = (\pi'_1, \cdots, \pi'_n)$ w.r.t. the makespan function, if the longest path in $\overline{\pi}$, denoted by $\pi_\text{max}$, is shorter than the longest path in $\overline{\pi}' $, denoted by $\pi_\text{max}'$. Therefore, we have that $\forall i\in [n]: |\pi_i|\leq |\pi_\text{max}| \leq |\pi_\text{max}'|$. 
That is: $\overline{\pi}$ is an optimal solution iff in all other solutions $\overline{\pi}'$, there exists a path $\pi_\text{max}'$ that is longer than all paths $\pi_i, i\in [n]$. We can easily express this constraint in HyperLTL. 
%
%
For example, the following formula states that $\overline{\pi}$ is an optimal solution for the 
standard MAPF problem (the supplementary appendix provides similar translations for all variants discussed in the paper):
\begin{align}
   & \label{line:sol1} \varphi_{\text{MAPF-makespan}} (\overline{\pi}):=    \varphi_{\text{MAPF}}(\overline{\pi}) \wedge \\ & \forall \pi_1' \ldots \forall \pi_n'. 
   ~\varphi_{\text{MAPF}}(\overline{\pi}') \rightarrow \label{line:sol2}
   \\&\bigvee_{i\in[n]} \bigwedge_{j\in [n]} (\neg\LTLglobally {g_i}_{\pi_i'}) \LTLuntil ( \LTLglobally {g_j}_{\pi_j}) \label{line:makespan}
\end{align}
\emph{Correctness.} \Cref{line:sol1} states that $\overline{\pi}$ is a solution. Equations~(\ref{line:sol2})-(\ref{line:makespan}) state that in every other solution, some agent reaches its goal no earlier than all agents in $\overline{\pi}$, which is exactly makespan optimality.

\paragraph{Using Quantified Propositions}
The direct formula above compares goal satisfaction times using the concrete MAPF goal predicate. A more uniform HyperQPTL encoding labels, for each agent, the unique timestep at which its goal is satisfied, and then compares these labels across candidate solutions. This turns the makespan comparison into a problem-independent suffix of the formula; only the local definition of ``goal satisfaction'' changes between MAPP variants. The supplementary appendix gives the full labeled translations, including the MAPF-PC case where the last occurrence of the sequence marker $q_i$ is used instead of the goal label.

\subsection{Sum-of-Costs}\label{sec:SOC}
SOC requires accumulating all path lengths rather than comparing the latest goal-satisfaction time. Although the labels used for makespan can identify each individual path length, their sum requires asynchronous temporal reasoning over positions of different traces. This is beyond synchronous HyperLTL/HyperQPTL, but can be expressed in asynchronous temporal hyperlogics, e.g.~\cite{DBLP:conf/spin/BombardelliBST25}, by constructing and comparing a trace that represents the accumulated path lengths.

\subsection{Optimal MAPP via HyperQPTL Model Checking}\label{sec:MC}
We now combine \Cref{sec:translation} and \Cref{sec:makespan} to
express the fact that an optimal solution exists. Again, we use the basic MAPF as an example, but we can instantiate the following formula with any problem of \Cref{sec:translation}. In addition, we can use any cost function that we are able to define in a decidable hyperlogic (e.g., SOC), instead of makespan. We define:  
\begin{align*}
\psi_{\text{opt}} := \exists \pi_1 \ldots \exists \pi_n. \varphi_{\text{MAPF-makespan}}
\end{align*}


\begin{theorem}\label{thm:optimal}
Let $T_\mathcal{G}$ be a transition system as described in \Cref{sec:lts}. Then:
 $T_\mathcal{G} \models \psi_{\text{opt}}$ if and only if there is an optimal solution to the given MAPP problem. 
 In particular, if $\pi_1, \ldots, \pi_n$ are witnesses to the outermost existential quantifiers in $\psi_{\text{opt}}$, then $\pi_1, \ldots, \pi_n$ form such an optimal solution. 
\end{theorem}

The proof of \Cref{thm:optimal} is straightforward given \Cref{thm:mapf} and the correctness of the corresponding formulas. 

\subsubsection{Quantifier Alternation}
 An important measure for the complexity of a HyperLTL or HyperQPTL formula is in the number of path-quantifier alternations: from existential to universal and vice versa.
 Each such alternation adds one exponent to the worst-case complexity of the model-checking problem. For example, the worst-case complexity of formulas of the type $\exists \pi \forall \pi'. \varphi$ where $\varphi$ is quantifier free, is exponentially harder than the worst case complexity of $\forall \pi'. \varphi$ formulas. Thus, we will always aim at formulas with minimal number of quantifier alternations. We observe that the modeling of almost all problems contains no path quantifiers within the formula, thus the only path quantification we introduce is the universal quantification in \Cref{sec:makespan} and the existential quantification of \Cref{thm:optimal}, resulting in an $\exists^*\forall^*$ formula, which is needed already for simpler problems than the ones we study here~\cite{BeutnerFink2025}.
We emphasize that reducing MAPP to HyperQPTL model checking is not an efficient solution given the current tools for hyperproperty verification, such as AutoHyper and related model checkers~\cite{DBLP:conf/tacas/BeutnerF23,DBLP:conf/lpar/BeutnerF23,DBLP:conf/tacas/HsuSB21,DBLP:conf/tacas/HsuSSB23}. This limitation is orthogonal to the modeling contribution: the formulas identify the exact hyperproperty that a solution must satisfy, while specialized MAPP algorithms remain the right baseline for large-scale performance. At the same time, the formulas isolate where complexity enters (for example, through quantifier alternation), which can guide future tool development for the reachability-heavy fragment used here.

\section{Related Work}\label{sec:related}

We consider a set of \emph{one-shot} MAPP variants, which end when the given tasks are fulfilled. However, \emph{online}~\cite{vsvancara2019online} and \emph{lifelong}~\cite{LifelongWan} variants introduce, arguably, the more realistic cases where new agents or tasks arrive over time, and a new solution must be computed. Our 
modeling can also be generalized for these cases by considering a sequence of one-shot problems.

In addition to search-based approaches, like CBS, there exist reduction-based approaches for optimally solving MAPP. These include reducing MAPP to \emph{Integer-Linear Program (ILP)}~\cite{yu2013multi}, Answer Set \emph{Programming (ASP)}~\cite{erdem2013general}, and \emph{SAT}~\cite{SurynekFSB16,PICAT_SOC}. Reducing MAPP to either of these problems, requires encoding temporal constraints, e.g., by adding many auxiliary variables. In contrast, HyperLTL and HyperQPTL offer a higher abstraction level, and natively express paths and the temporal constraints. 

Many works have studied the relation between planning and temporal properties. 
\cite{DBLP:conf/ifaamas/PontiggiaMAC025} study the connection between probabilistic hyperproperties and decentralized planning. While their work is concerned with undecidable problems in the probabilistic setting, 
we focus on the Boolean setting, 
and put more emphasis on modeling concrete MAPP problems, and, on studying a wider variety of problems. We show that we can express most problems in a modular way, as a conjunction of 
formulas, 
thus one can easily adjust only the relevant part(s) when introducing new problems.
\cite{icra20} study the connection between HyperLTL and robotics, and \cite{BeutnerFink2025} study the connection between hyperproperties and conformant planning.  While they discuss hyperproperties and planning, our work studies \emph{multi-agent} path planning. 
%
%
\cite{DBLP:conf/concur/BeutnerF21,DBLP:journals/lmcs/BeutnerF23}
present a hyperlogic to reason about multi-agent systems, in which they express properties such as information-flow and security properties, over multi-agent systems. However, they do not address real-life planning problems as we study here.

In contrast to all of the above, 
in this work, we study a different aspect of the relation between hyperlogics and multi-agent systems. We are interested in \emph{modeling existing problems from the planning community}, in an existing hyperlogic, thus establishing the tight connection between the two communities. We stress that all of the above works do not apply their results on actual planning problems, that are widely studied in the planning community. 
%
%
In addition, none of the above works concern cost functions and optimality of the solutions, as we discuss in \Cref{sec:cost}. 
This distinguishes our work from empirical MAPP papers: our comparison point is not the runtime of a new solver, but the expressive match between MAPP variants and temporal hyperproperties.

\section{Conclusion and Future Work}
In this paper, we have established the fundamental connection between \emph{Multi-Agent Path Planning} (MAPP) problems and hyperlogics. Modeling the different problems in a formal language helps expose the subtleties of the variants, remove accidental differences between formulations, and specify future problems using a clean and unambiguous formalism. The value of this view is therefore conceptual and methodological: it gives MAPP researchers a common specification language and gives the hyperproperties community a family of concrete planning problems with rich cross-trace structure.
A specific domain in which our study can be applicable is explainability and causality, where we plan to apply explainability techniques for hyperproperties \cite{hypervis,coenen2022hyperproperty} to the setting of MAPP. A MAPP solution is naturally represented by several witness traces, and the formulas in this paper state which relations among those traces make the plan valid. This creates a direct interface to formal explanations: one can ask which trace relations caused a conflict to be avoided, a goal to be delayed, or a group constraint to hold. 
In addition, we plan to study further constraints on the planning problems (as mentioned briefly in \Cref{sec:motivation}), and combine them with verification tasks, to use the full power of hyper-temporal specifications.





\bibliography{sample,aaai2027}


\newpage
\appendix
\input{appendix_for_arxiv}
\end{document}

%% file: macros.tex
\newcommand{\TR}{\text{TRACES}}
\newcommand{\tuple}[1]{\ensuremath{\left \langle #1 \right \rangle }}

\newtheorem{theorem}{Theorem}
\newtheorem{example}{Example}
\newtheorem{observation}{Observation}

%% file: appendix_for_arxiv.tex
\section*{Appendix}
\vspace{0.2cm}
\section{Full Translations for the Makespan Cost Function}
In this appendix, we provide the encoding of the makespan function to all variants we consider in the paper. We begin with the version that does not use quantified propositions, and then address the version with quantified propositions. 
\subsection{Direct Translation Using HyperLTL}\label{app:hyperltl}
 In this section, we cannot use the goal formula as a black-box, as we explain in \Cref{sec:makespan}. 
The translations that use quantified propositions
can use the goal function (and the whole formula describing the problem) as a black box, with a minor modification to the original formulas. In the formulas we present here, we often need to reason about the goal formula, and to define optimality with explicit reasoning over goal formulas. 
\paragraph{(1) MAPF}
We have shown the modeling of makespan for the MAPF problem in the body of the paper (\Cref{sec:makespan}), but we provide it here again for completeness. 
\begin{align*}
   &  \varphi_{\text{MAPF-makespan}} (\overline{\pi}):=    \varphi_{\text{MAPF}}(\overline{\pi}) \wedge \\ & \forall \pi_1' \ldots \forall \pi_n'. 
   ~\varphi_{\text{MAPF}}(\overline{\pi}') \rightarrow
   \\&\bigvee_{i\in[n]} \bigwedge_{j\in [n]} (\neg\LTLglobally {g_i}_{\pi_i'}) \LTLuntil ( \LTLglobally {g_j}_{\pi_j})
\end{align*}

\paragraph{(2) AMAPF}
Recall that:
\begin{align*}
    \varphi_{\text{AMAPF}}(\overline{\pi}) := &
  \exists \pi_1' \ldots \exists \pi_n'. ~ \varphi_{\text{start}} (\overline{\pi}) \wedge \varphi_{\text{goal}} (\overline{\pi'}) ~\wedge
\\
&  \varphi_{\text{no-conf}}(\overline{\pi}) \wedge 
\varphi_{\text{no-conf}}(\overline{\pi}') ~\wedge
\\
  & \forall \tau .\exists \tau'. \bigwedge_{i\in[n]} \Big( (\tau = \pi_i ) \rightarrow\\ & \bigvee_{j\in [n]} \big( (\tau' = \pi_j' ) \wedge (\tau = \tau')\big) \Big)
\end{align*}
Here, we need to explicitly reason about the existentially quantified traces $\pi'$ in $ \varphi_{\text{AMAPF}}(\overline{\pi})$, and where they reach their goal vertices, and not only the $\pi$ paths. We use $\overline{\tau}$ to denote non-optimal solutions (and respectively, let $\overline{\tau}'$ be the existentially quantified paths in $\varphi_{\text{AMAPF}}(\overline{\tau})$). 
\begin{align}
     & \varphi_{\text{AMAPF-makespan}} := \varphi_{\text{AMAPF}}(\overline{\pi}) \wedge \\
     & \forall \tau_1 \ldots \forall \tau_n. \varphi_{\text{AMAPF}}(\overline{\tau}) \rightarrow   \\
     &\bigvee_{i\in[n]} \bigwedge_{j\in [n]} (\neg\LTLglobally {g_i}_{\tau_i'}) \LTLuntil ( \LTLglobally {g_j}_{\pi_j'}) \label{eq:makespanAMAPF}
\end{align}
In \Cref{eq:makespanAMAPF} we compare the quantified path $\tau_i'$ to all quantified paths $\pi_j'$, since we use these quantified paths in $\varphi_{\text{AMAPF}}$ in the goal formula, and thus we want to make sure that there is a path  $\tau_i'$ that is longer than all paths $\pi_j'$. 

\paragraph{(3) MCPF} 
We recall the goal formula for MCPF, where the goals are not preassigned. 
$$\varphi_{\text{goal-MCPF}}(\overline{\pi}) := \bigwedge_{j\in[m]} \bigvee_{i\in[n]} \LTLeventually  {g_j}_{\pi_i} $$
Then, in a non-optimal solution $\overline{\pi}'$ there exists
a goal vertex $g_l$ such that for every agent $a_k$, either $a_k$ does not reach $g_l$, or it reaches it after some agent has reached it in the solution~
$\overline{\pi}$. 

Therefore: 
\begin{align*}
   &  \varphi_{\text{MCPF-makespan}} (\overline{\pi}):=    
    \varphi_{\text{MCPF}} (\overline{\pi})
  \wedge \\
   & \forall \pi_1' \ldots \forall \pi_n'. 
   ~\varphi_{\text{MAPF}}(\overline{\pi}') \rightarrow
   \\&\bigvee_{l\in[m]} \bigwedge_{k\in [n]} (\neg {g_l}_{\pi_k'}) \LTLuntil  (
   \bigwedge_{j\in [m]}  \bigvee_{i\in [n]} 
   {g_j}_{\pi_i}) 
\end{align*}

\paragraph{(4) TAPF} 
Here, each agent has several goals, and it must reach one of them. 
Then, in a non-optimal solution, 
there is an agent $a_k$ that does not reach \emph{any} of its goals until all agents in the optimal solution have reached some goal. 
Recall that:
$$ \varphi_{\text{goal-TAPF}} (\overline{\pi})
:= \bigwedge_{i\in[n]} \bigvee_{g \in gset_i} \LTLeventually \LTLglobally {g}_{\pi_i}$$
Then, we define: 
\begin{align*}
   &  \varphi_{\text{TAPF-makespan}} (\overline{\pi}):=    
    \varphi_{\text{start}} (\overline{\pi})
  \wedge \varphi_{\text{no-conf}} (\overline{\pi}) 
   \wedge \\ 
& \bigwedge_{i\in[n]} \bigvee_{g \in gset_i} \Big( \LTLeventually \LTLglobally {g}_{\pi_i} \wedge \\
   %
   & \forall \pi_1' \ldots \forall \pi_n'. 
   ~\varphi_{\text{MAPF}}(\overline{\pi}') \rightarrow
   \\&\bigvee_{k\in[n]} \bigwedge_{g' \in gset_k} (\neg\LTLglobally {g_l}_{\pi_k'}) \LTLuntil (  \LTLglobally{g}_{\pi_i}) \Big)
\end{align*}
The formula $ \varphi_{\text{TAPF-makespan}} $ is not in proper HyperLTL syntax, but we can bring it to proper syntax by moving the universal quantifiers outside (this does not affect the $\pi$ traces as they are not under the scope of this quantification).
\paragraph{(5) MAPD} 
In MAPD, we have a set of pickup-delivery assignments, each should be performed by some agent. Then, similarly to the MCPF, we have defined the goal formulas 
\begin{equation*}
    \varphi_\text{goal-MAPD}:= \bigwedge_{j\in[m]}\bigvee_{i\in[n]}\LTLeventually ({p_j}_{\pi_i} \wedge \LTLeventually {d_j}_{\pi_i})
\end{equation*}
In particular, in a solution $\overline{\pi}'$ which is not optimal, there is a delivery vertex $d_l$ and an agent $a_k$ that reaches $d_l$ after all delivery vertices have been reached in $\overline{\pi}$. Thus, it is enough to reason about the delivery points, and we get a formula very similar to this of the MCPF problem:
\begin{align*}
   &  \varphi_{\text{MAPD-makespan}} (\overline{\pi}):=    
    \varphi_{\text{MAPD}} (\overline{\pi})
  \wedge \\
   & \forall \pi_1' \ldots \forall \pi_n'. 
   ~\varphi_{\text{MAPF}}(\overline{\pi}') \rightarrow
   \\&\bigvee_{l\in[m]} \bigwedge_{k\in [n]} (\neg {d_l}_{\pi_k'}) \LTLuntil  (
 \bigwedge_{j\in [m]}  \bigvee_{i\in [n]} 
   {d_j}_{\pi_i}) 
\end{align*}

\paragraph{(6) MAPF-PC} 
If $\overline{\pi}'$ is a solution to the MAPF-PC, but it is not optimal, then we have that some agent $a_k$ in $\overline{\pi}'$ satisfies its last goal $g_k^{m_k}$ after all other agents have satisfied their last goal. Recall that we have also used the proposition $q_k$ to mark the timesteps on which the agents visit a vertex to satisfy a goal. 
Then, while the definition of MAPF-PC in \Cref{sec:translation} is quiet elaborated, to reason about optimality we only need to reason about the last goal vertex $g_k^{m_k}$ and the proposition $q_k$, and we get a formula very similar to this of the MAPF problem: 
\begin{align*}
   &  \varphi_{\text{MAPF-PC-makespan}} (\overline{\pi}):=    \varphi_{\text{MAPF}}(\overline{\pi}) \wedge \\ & \forall \pi_1' \ldots \forall \pi_n'. 
   ~\varphi_{\text{MAPF-PC}}(\overline{\pi}') \rightarrow
   \\&\bigvee_{k\in[n]} \bigwedge_{j\in [n]} (\neg({g_k^{m_k}}_{\pi_k'}\wedge q_k ) ) \LTLuntil ( {g_j^{m_j}}_{\pi_j} \wedge q_j)
\end{align*}

\paragraph{(7) RMAPF}
Here, the goal formula remains the same as in the MAPF problem, thus we have the same cost function as for the MAPF problem, when changing 
$  \varphi_{\text{MAPF}}$ to   
$\varphi_{\text{RMAPF}}$: 
\begin{align*}
   &  \varphi_{\text{RMAPF-makespan}} (\overline{\pi}):=    \varphi_{\text{RMAPF}}(\overline{\pi}) \wedge \\ & \forall \pi_1' \ldots \forall \pi_n'. 
   ~\varphi_{\text{RMAPF}}(\overline{\pi}') \rightarrow
   \\&\bigvee_{i\in[n]} \bigwedge_{j\in [n]} (\neg\LTLglobally {g_i}_{\pi_i'}) \LTLuntil ( \LTLglobally {g_j}_{\pi_j})
\end{align*}

\paragraph{(8) MTPF} For the multi-train path finding problem, we again use the simple goal formula, and simply define: 
\begin{align*}
   &  \varphi_{\text{MTPF-makespan}} (\overline{\pi}):=    \varphi_{\text{MTPF}}(\overline{\pi}) \wedge \\ & \forall \pi_1' \ldots \forall \pi_n'. 
   ~\varphi_{\text{MTPF}}(\overline{\pi}') \rightarrow
   \\&\bigvee_{i\in[n]} \bigwedge_{j\in [n]} (\neg\LTLglobally {g_i}_{\pi_i'}) \LTLuntil ( \LTLglobally {g_j}_{\pi_j})
\end{align*}
This demonstrates the modularity of the HyperLTL modeling: We can add many additional constraints to the problem, but as long as we keep the goal characterization relatively similar to previous problems, we can use the similar cost function. 

\paragraph{(9) MAM} Recall that 
\begin{align*}
  &  \varphi_{\text{MAM}}(\overline{\pi}) :=  
    ~\varphi_{\text{start}}(\overline{\pi}) ~\wedge ~ \\ &\bigvee_{v\in AP} \Big(\big( \LTLeventually \bigwedge_{i\in [n]} v_{\pi_i} \big) ~\wedge~  \\ &
     \LTLglobally
    \big( \bigwedge_{i\neq j} \bigwedge_{\substack{v'\in AP \\ v' \neq v}} 
    \neg (v'_{\pi_i} \wedge v'_{\pi_j}) \big) \Big) 
\end{align*}
That is, there exists a vertex $v$ which is the meeting point of all agents. Thus, in a non-optimal solution $\overline{\pi}'$, for every possible meeting vertex $v'$, it cannot be that all agents are in $v'$ using the solution $\overline{\pi}'$, until all of them meet in $v$ using the optimal solution  $\overline{\pi}$. 
%
We explicitly reason about the meeting vertex $v$: 
\begin{align*}
  & \varphi_{\text{MAM-makespan}} (\overline{\pi}) := 
  \varphi_{\text{MAM}}(\overline{\pi}) \wedge \\ 
& \forall \pi_1' \ldots \forall \pi_n'. 
   ~\varphi_{\text{MAM}}(\overline{\pi}') \rightarrow
  ~ \\ &\bigwedge_{v'\in AP} 
  (\neg \bigwedge_{i\in [n]} v'_{\pi'_i} ) \LTLuntil ( \bigwedge_{i\in [n]} v_{\pi_i} ) 
\end{align*}


\subsection{Translation Using HyperQPTL and Quantified Propositions}\label{app:hyperqptl}
The use of HyperQPTL is much cleaner than the use of HyperLTL for cost functions. Here, we do not need to reason about the goal, nor understand it. We simply add quantified propositions that mark the point in time in which an agent has reached its goal. 


We use the formula $ \varphi_{\text{once}_i}$, that states that the (quantified) proposition $l_i$ holds exactly once. $ \varphi_{\text{once}_i}$ will be used in the translation throughout this section. 
\begin{align*}
    \varphi_{\text{once}_i}:=  (\neg l_i) \LTLuntil (l_i \wedge \LTLnext \LTLglobally \neg l_i )
\end{align*}
In the following, we only need to show how we modify the goal formula for each problem (except for MAPF-PC that does not have a clear goal formula), to include the quantified propositions. 
For MAPF-PC, we label the last goal, as discussed in \Cref{sec:makespan}.

\paragraph{(1) MAPF} 
We define the labeled goal formula as: 
\begin{align*}
   & \varphi_{\text{L-goal}}(\overline{\pi}) := \exists l_1 \cdots \exists l_n.\bigwedge_{i\in[n]} \LTLeventually \big( (\LTLglobally  {g_i}_{\pi_i} ) \wedge {l_i}_{\pi_i} \big)\wedge 
   \bigwedge_{i\in[n]}   \varphi_{\text{once}_i}
\end{align*}
We then use the labeled formula $\varphi_{\text{L-goal}}$ instead of the original $\varphi_{\text{goal}}$ formula in \Cref{sec:translation}, resulting in the labeled formula $\varphi_{\text{L-MAPF}}$ for the MAPF problem: 
\begin{align*}
      \varphi_{\text{L-MAPF}}(\overline{\pi}):=
  \varphi_{\text{start}}(\overline{\pi}) \wedge \varphi_{\text{L-goal}}(\overline{\pi}) \wedge \varphi_{\text{no-conf}}(\overline{\pi})
\end{align*}
The only change is in using $\varphi_{\text{L-goal}}$ instead of $\varphi_{\text{goal}}$. 
We use $\varphi_{\text{L-goal}}$ in all formulas in which we have used $\varphi_{\text{goal}}$, as we show in the following.

\paragraph{(2) AMAPF} 
Since $ \varphi_{\text{AMAPF}}$ uses $\varphi_{\text{goal}}$, we can simply define $ \varphi_{\text{L-AMAPF}}$ by using  $\varphi_{\text{L-goal}}$. We note the equality over traces implies equality over both quantified and not quantified propositions. 
\begin{align*}
    \varphi_{\text{L-AMAPF}}(\overline{\pi}) := &
  \exists \pi_1' \ldots \exists \pi_n'. ~ \varphi_{\text{start}} (\overline{\pi}) \wedge \varphi_{\text{L-goal}} (\overline{\pi'}) ~\wedge
\\
&  \varphi_{\text{no-conf}}(\overline{\pi}) \wedge 
\varphi_{\text{no-conf}}(\overline{\pi'}) ~\wedge
\\
  & \forall \tau .\exists \tau'. \bigwedge_{i\in[n]} \Big( (\tau = \pi_i ) \rightarrow \\ & \bigvee_{j\in [n]} \big( (\tau' = \pi_j' ) \wedge (\tau = \tau')\big) \Big)
\end{align*}

\paragraph{(3) MCPF} Here, we need to add quantified propositions for the satisfaction of each goal, and not necessarily for each agent. In fact, we do not need to do this at all, as we can directly use the goal vertices (as, once an agent has reached the vertex, it has satisfied the goal), but we use the quantified propositions for consistency. Since in the formula $  \varphi_{\text{L-MAPF-makespan}} $ we reason about quantified propositions related to the agents, and not the goals, we manipulate the formula a little bit, such that a label $l_i$ holds for an agent $a_i$ at the last timestep in which the agent has satisfied a goal. Note that an agent can satisfy multiple goals, and then $l_i$ will label the last one; and it can satisfy no goals, and then $l_i$ will hold at the first timestep. 
\begin{align}
    &\varphi_{\text{L-goal-MCPF}}(\overline{\pi}) := 
\exists l_1 \cdots \exists l_n. 
   \bigwedge_{i\in[n]}   \varphi_{\text{once}_i} \wedge \\
&\bigwedge_{j\in[m]} \bigvee_{i\in[n]} \LTLeventually  ({g_j}_{\pi_i} \wedge l_i) ~\wedge \\ 
& \LTLglobally \Big( \bigwedge_{i\in[n]} \bigwedge_{j\in[m]} \big( ({g_j}_{\pi_i} \wedge l_i) \rightarrow \neg \LTLeventually \bigvee_{l\in[m]  } {g_l}_{\pi_i} \big) 
\Big) \label{eq:lastpos} ~\wedge \\
& \bigwedge_{i\in[n]}\big( (\LTLglobally \bigwedge_{j\in[m]} \neg  {g_j}_{\pi_i}) \rightarrow l_i    \big) \label{eq:firstpos}
\end{align}
\Cref{eq:lastpos} states that for an agent $a_i$, if the label $l_i$ holds when some goal $g_j$ is reached, then it is the last goal of agent $a_i$: There is no other goal vertex $g_l$ that is reached later on the trace $\pi_i$ of agent $a_i$. \Cref{eq:firstpos} states, for every agent $a_i$, that if no goal it ever satisfied on $\pi_i$, that is, agent $a_i$ satisfies no goal, then $l_i$ holds at the first position (under the scope of no temporal operators).
 
Then, we define: 
\begin{align*}
    \varphi_{\text{L-MCPF}}(\overline{\pi}) :=  
  \varphi_{\text{start}} (\overline{\pi}) \wedge \varphi_{\text{L-goal-MCPF}}(\overline{\pi}) 
  \wedge \varphi_{\text{no-conf}} (\overline{\pi}) 
\end{align*}

\paragraph{(4) TAPF} 
Here, each agent has several goals, and it must reach one of them, to which we add the label $l_i$:
\begin{align*}
&\varphi_{\text{L-goal-TAPF}} (\overline{\pi})
:= 
\exists l_1 \cdots \exists l_n. 
   \bigwedge_{i\in[n]}   \varphi_{\text{once}_i} \wedge \\
&\bigwedge_{i\in[n]} \bigvee_{g \in gset_i} \LTLeventually ((\LTLglobally {g}_{\pi_i}) \wedge l_i)
\end{align*} 
And then:
\begin{equation*}
    \varphi_{\text{L-TAPF}}(\overline{\pi}) :=  
  \varphi_{\text{start}} (\overline{\pi})\wedge \varphi_{\text{L-goal-TAPF}}(\overline{\pi})
  \wedge \varphi_{\text{no-conf}} (\overline{\pi})
\end{equation*}

\paragraph{(5) MAPD} 
In MAPD, we have a set of pickup-delivery assignments, each should be performed by some agent.
Again, an agent can preform no assignment, or several assignments. Thus, we use the same reasoning as in the MCPF problem, with the modification of taking $d_j$ instead of $g_j$:
\begin{align*}
    &\varphi_\text{L-goal-MAPD}(\overline{\pi}):= 
    \exists l_1 \cdots \exists l_n. 
   \bigwedge_{i\in[n]}   \varphi_{\text{once}_i} \wedge \\
   & \bigwedge_{j\in[m]}\bigvee_{i\in[n]}\LTLeventually ({p_j}_{\pi_i} \wedge \LTLeventually ({d_j}_{\pi_i}\wedge l_i))
    \\ 
& \LTLglobally \Big( \bigwedge_{i\in[n]} \bigwedge_{j\in[m]} \big( ({d_j}_{\pi_i} \wedge l_i) \rightarrow \neg \LTLeventually \bigvee_{l\in[m]  } {d_l}_{\pi_i} \big) 
\Big)  ~\wedge \\
& \bigwedge_{i\in[n]}\big( (\LTLglobally \bigwedge_{j\in[m]} \neg  {d_j}_{\pi_i}) \rightarrow l_i    \big) 
\end{align*}
Then, we define:
\begin{equation*}
    \varphi_{\text{L-MAPD}}(\overline{\pi}) :=  
    \varphi_{\text{start}} (\overline{\pi}) \wedge \varphi_\text{L-goal-MAPD}(\overline{\pi})   \wedge \varphi_{\text{no-conf}} (\overline{\pi})
\end{equation*}

\paragraph{(6) MAPF-PC} 
Recall that for the MAPF-PC problem, we have already used quantified propositions, to label the $m_i$ positions in which an agent satisfies a goal, using the formula $\psi_i$. Thus, we can simply label the last such position, by slightly modifying $\psi_i$, and adding a conjunct with $l_i$ at the last appearance of $q_i$ (end of \Cref{eq:pcaddli}):
\begin{align}
  &\psi_{L,i}:=   \varphi_{\text{once}_i} ~\wedge  \\
  \label{eq:pcaddli}&\overbrace{  \neg q_i \LTLuntil (q_i \wedge \LTLnext (\neg q_i \LTLuntil ( q_i \wedge \LTLnext (\cdots  \wedge \LTLnext(\neg q_i \LTLuntil ((q_i\wedge l_i) }^{\neg q_i\LTLuntil q_i \text{ appears }  m_i \text{ times}}\\ &\wedge \LTLnext \LTLglobally \neg q_i)\cdots
 ) 
\end{align}
For the complete $\varphi_{\text{L-MAPF-PC}}(\overline{\pi})$ formula, we use the formulas
$\psi_{L,i}$ for $i\in [n]$, and the formulas
$\varphi_{\text{gseq}}$ and $   \varphi_{(g_i^x, g_j^y)}$ as defined in \Cref{sec:translation}:
\begin{align*}
    &\varphi_{\text{L-MAPF-PC}}(\overline{\pi}) := \exists q_1 \ldots \exists q_n. \exists l_1 \ldots \exists l_n. ~ 
    \bigwedge_{i\in [n]} \psi_{L,i} ~\wedge \\ & \varphi_{\text{no-conf}}  (\overline{\pi})~\wedge 
     \bigwedge_{i\in [n]} \varphi_{\text{gseq}}(\pi_i) ~\wedge
    \bigwedge_{(g_i^x, g_j^y) \in PC}  \varphi_{(g_i^x, g_j^y)}(\pi_i, \pi_j)
\end{align*}

\paragraph{(7) RMAPF}
The formula for the robust MAPF uses the same goal formula as the MAPF problem, thus we can simply use $  \varphi_{\text{L-goal}} $ instead: 
\begin{align*}
        \varphi_{\text{L-RMAPF}} (\overline{\pi}) :=  \varphi_{\text{start}} (\overline{\pi})  ~\wedge~     \varphi_{\text{L-goal}} (\overline{\pi})~ \wedge  ~ \varphi_{\text{robust-no-conf}} (\overline{\pi}) 
\end{align*}

\paragraph{(8) MTPF}
Similarly, the multi-train path finding problem also uses the same goal formula, thus:
\begin{align*}
    \varphi_{\text{L-MTPF}}  (\overline{\pi}): = 
    \varphi_{\text{start}} (\overline{\pi})  ~\wedge~     \varphi_{\text{L-goal}} (\overline{\pi})~ \wedge  ~ 
    \varphi_{\text{MTPF-no-conf}} (\overline{\pi})
\end{align*}

\paragraph{(9) MAM}
In the meeting variant, we actually need only one proposition, that marks the point in time in which all agents arrive at the meeting vertex $v$. Again, for consistency, and for easier definition of the makespan formula, we use a quantified proposition $l_i$ for each agent $a_i$, and require that they all hold only once (\Cref{eq:mamonce}), and at the same time (\Cref{eq:mamtime}). Here we do not have a goal formula, so we modify the whole formula $ \varphi_{\text{MAM}}$: 

\begin{align}
  &  \varphi_{\text{L-MAM}}(\overline{\pi}) :=  \exists l_1 \cdots \exists l_n. \label{eq:mamonce}
   \bigwedge_{i\in[n]}   \varphi_{\text{once}_i} \wedge \\ 
    &\varphi_{\text{start}}(\overline{\pi}) ~\wedge ~ \\ &\bigvee_{v\in AP} \Big( \LTLeventually \label{eq:mamtime}
    \big( \bigwedge_{i\in [n]} v_{\pi_i} \wedge \bigwedge_{i\in [n]} l_i \big) ~\wedge~  \\ &
     \LTLglobally
    \big( \bigwedge_{i\neq j} \bigwedge_{\substack{v'\in AP \\ v' \neq v}} 
    \neg (v'_{\pi_i} \wedge v'_{\pi_j}) \big) \Big) 
\end{align}


\subsubsection{SUMMARY: Expressing the Makespan Cost Function Using HyperQPTL and Quantified Propositions}
In the above, we have defined new labeled formulas for each of the nine variants, such that the label $l_i$ holds at timestep $t$ iff agent $a_i$ satisfied (all of) its goal(s) at timestep $t$. That is, the timestep in which $l_i$ holds in a solution $\overline{\pi} = (\pi_1, \cdots, \pi_n)$ is exactly the length (and cost) of the path $\pi_i$. Thus, to make sure that
$ \overline{\pi}$ is more optimal than  $\overline{\pi}'$, we need the check whether there is an agent $a_k$ for which $l_k$ holds in  $\overline{\pi}'$ after all $l_i$'s occurred in $\overline{\pi}$. Here, we can use a generic formula $   \varphi_\text{L} $ that can be instantiated with any of the nine formulas above, that is: 
\begin{align*}
    \varphi_\text{L} \in \{  \varphi_{\text{L-MAPF}},   \varphi_{\text{L-AMAPF}},  \varphi_{\text{L-MCPF}},  \varphi_{\text{L-TAPF}}, 
  \varphi_{\text{L-MAPD}}, \\ \varphi_{\text{L-MAPF-PC}},
   \varphi_{\text{L-RMAPF}}, \varphi_{\text{L-MTPF}},  \varphi_{\text{L-MAM}}
\}
\end{align*}
For a (non-optimal) solution $\overline{\pi}'$, we denote its quantified propositions by $l_i', i\in[n]$. 
We then define optimality with respect to the makespan function as follows:
\begin{align*}
     &\varphi_{\text{L-makespan}} (\overline{\pi}):=    \varphi_{\text{L}}(\overline{\pi}) \wedge \\ & \forall \pi_1' \ldots \forall \pi_n'. \varphi_{\text{L}}(\overline{\pi}') \rightarrow \bigvee_{i\in[n]} \bigwedge_{j\in [n]}(\neg {l'_i}_{\pi_i'}) \LTLuntil {l_j}_{\pi_j}
\end{align*}